%% file: main.tex
\documentclass[runningheads]{llncs}
\input{packages.tex}

\input{mysetting_article.tex}

\input{definitions.tex}
\usepackage[T1]{fontenc}
\begin{document}
\title{Effectful Toposes and Their Lawvere-Tierney Topologies}
\author{Rinta YAMADA}
\author{Rinta Yamada\inst{1,2} \orcidlink{0009-0009-5927-8684}}

\authorrunning{R. Yamada}
\institute{The University of Electro-Communications, Chofu, Tokyo, Japan
\and Keio University, Minato, Tokyo, Japan \\
\email{rinta.yamada.cs@gmail.com}}
\maketitle              
\begin{abstract}
This paper introduces {effectful toposes} as an extension of the effective topos and investigates their structure relative to {Lawvere-Tierney topologies}.
First, we formulate effectful toposes by lifting the evidenced frame, which is a recently proposed model for effectful computation.
Next, we define Lawvere-Tierney topologies on effectful toposes and characterize the sheaves on them.
In the effective topos, Lawvere-Tierney topologies are known to correspond to computational oracles.
Finally, to demonstrate that our result contributes to the connection between realizability relativized by oracles and effects,
we show that the sheaf topos for the double negation topology is isomorphic to the effectful topos with the Continuation-Passing-Style effect;
these toposes serve as a model of classical realizability.

\end{abstract}

\input{section/introduction.tex}

\input{section/ef_to_topos.tex}

\input{section/mca_k1.tex}

\input{section/topology.tex}

\input{section/conclusion.tex}

\begin{credits}
\subsubsection{\ackname} 
I would like to thank my supervisor, Prof.\ Yoshio Okamoto, for his guidance and support.
I also thank Assoc.\ Prof.\ Ryota Akiyoshi for his helpful comments on the manuscript.

\subsubsection{\discintname}
The author has no competing interests to declare that are relevant to the content of this article.
\end{credits}

\newpage

\appendix
\renewcommand{\thesection}{Appendix \arabic{section}}
\input{section/appendix.tex}

%
%
%
\bibliographystyle{splncs04}
\bibliography{etltt}

\end{document}

%% file: packages.tex
\usepackage{orcidlink}
\usepackage{amsmath,amsfonts,amssymb,bm,ascmac}
\usepackage{physics}
\usepackage{tikz}
\usetikzlibrary{intersections,calc,arrows.meta,cd,automata,positioning,shapes.geometric}
\usepackage{braket,enumerate,mathrsfs,stmaryrd}
\usepackage{url}
\usepackage[normalem]{ulem}
\usepackage{listings}
\usepackage{bussproofs}
\usepackage{mathtools}
\usepackage{cite}
\usepackage{docmute}
\usepackage{upgreek}
\usepackage{enumitem}
\usepackage{bbm}
\usepackage{caption}
\usepackage{color}

\usepackage{hyperref}
\hypersetup{%
 setpagesize=false,%
 bookmarks=true,%
 bookmarksdepth=tocdepth,%
 bookmarksnumbered=true,%
 colorlinks=true,%
 linkcolor=blue,
 citecolor=blue,
 urlcolor=blue,
 pdftitle={},%
 pdfsubject={},%
 pdfauthor={},%
 pdfkeywords={}}


%% file: mysetting_article.tex
\DeclareMathOperator{\Natural}{{\mathbb{N}}}

\newcommand{\id}{\mathrm{id}}

\newcommand{\CSet}{{\mathord{\mathbf{Set}}}}

\DeclarePairedDelimiter{\encode}{\langle}{\rangle}

\newcounter{general}[section]

%% file: definitions.tex
\spnewtheorem{construction}[definition]{Construction}{\bfseries}{\itshape}

\newcommand{\orthog}{\mathbin{\bot}}
\newcommand{\kexps}{{E(\Natural)}}
\newcommand{\EFkleene}{\EF_{\kleenek}}
\newcommand{\Char}{\mathord{\mathrm{Ch}}}
\newcommand{\Sheaf}{\mathop{\mathrm{Sh}}}

\newcommand{\pinc}{\mathop{\mathtt{inc}}}
\newcommand{\pidm}{\mathop{\mathtt{idm}}}
\newcommand{\pprs}{\mathop{\mathtt{prs}}}

\newcommand{\psep}{\mathtt{sep}}

\newcommand{\pdsc}{\mathtt{dsc}}

\newcommand{\psym}{\mathop{\mathtt{sym}}}

\newcommand{\ptrs}{\mathop{\mathtt{trs}}}

\newcommand{\pext}{\mathop{\mathtt{ext}}}

\newcommand{\pfext}{\mathop{\mathtt{fext}}}

\newcommand{\psv}{\mathop{\mathtt{sv}}}

\newcommand{\ptot}{\mathop{\mathtt{tot}}}

\newcommand{\peq}{\mathop{\mathtt{eq}}}

\newcommand{\pleq}{\mathop{\mathtt{leq}}}

\newcommand{\EffHOL}{{\mathord{\mathsf{EffHOL}}}}
\newcommand{\EF}{\mathord{\mathcal{E\!F}}}

\newcommand{\evrel}[1]{\mathrel{\xrightarrow{#1}}}

\newcommand{\lift}[1]{{\overline{#1}}}

\newcommand{\modalangle}[3]{ {\langle \!\!\mathrel{\raisebox{-4pt}{$\diamond$}}}\, #1 \leftarrow #2  {\mathrel{\raisebox{-4pt}{$\diamond$}}\!\!\rangle}\,   #3 }
\newcommand{\makleene}[3]{ {\langle \!\!\mathrel{\raisebox{-4pt}{$\diamond$}}}\, #1 \leftarrow #2  {\mathrel{\raisebox{-4pt}{$\diamond$}}\!\!\rangle}_{\kleenek}\,   #3 }
\newcommand{\macps}[3]{ {\langle \!\!\mathrel{\raisebox{-4pt}{$\diamond$}}}\, #1 \leftarrow #2  {\mathrel{\raisebox{-4pt}{$\diamond$}}\!\!\rangle}_{\cps}\,   #3 }

\newcommand{\UFam}{\mathsf{UFam}}

\newcommand{\efequiv}{\leftrightsquigarrow}
\newcommand{\efequivop}{\mathop{\leftrightsquigarrow}}

\DeclarePairedDelimiter{\eqclass}{\lbrack}{\rbrack}
\newcommand{\rsex}{\mathsf{ex}}

\newcommand{\singleton}{{\mathord{\mathbbm 1}}}
\newcommand{\boolhyt}{{\mathord{\mathbbm 2}}}

\newcommand{\true}{\mathop{\mathsf{true}}}
\newcommand{\tripc}[1]{\mathsf{hold}_{#1}}

\newcommand{\cps}{{\mathsf{CPS}}}
\newcommand{\dnt}{{\lnot\lnot}}
\newcommand{\algebracps}{{\mathcal K_\cps}}
\newcommand{\codeS}{e_{\mathtt{S}}}
\newcommand{\codeK}{e_{\mathtt{K}}}
\newcommand{\ethrow}[1]{e_{\mathsf{K}_{#1}}}
\newcommand{\pden}{\mathtt{den}}

\newcommand{\algebraA}{{\mathord{\mathbb{A}}}}
\newcommand{\join}{\mathbin{\sqcup}}

\newcommand{\meet}{\mathbin{\sqcap}}
\newcommand{\bigmeet}{\mathop{\bigsqcap}}
\newcommand{\himp}{\mathbin{\sqsupset}}
\newcommand{\modality}{\mathord{\diamondsuit}}
\newcommand{\htop}{\mathord{\bm{1}}}
\newcommand{\hbot}{\mathord{\bm{0}}}

\newcommand{\kleenek}{{\mathord{\mathcal{K}_1}}}
\newcommand{\effomega}{\boolhyt^{\Natural}}

\newcommand{\CatHp}{\mathbf{Hp}}
\newcommand{\tripP}{\mathsf{P}}

\newcommand{\tripET}{\mathsf{ET}}

\newcommand{\toposE}{\mathcal{E}}

\newcommand{\eftopos}{\mathrm{EFT}}
\newcommand{\toposEff}{\mathord{\eftopos(\EFkleene)}}

\def\BaseUrl{https://dominik-kirst.github.io/mca/effhol/}
\newcommand{\coqdoc}[1]{\href{\BaseUrl#1}{\raisebox{-.9mm}{\includegraphics[height=1em]{coql.png}}}}

\DeclareMathAlphabet{\sys}{OT1}{cmss}{sbc}{n}

\newcommand{\separator}{\mathcal{S}}

\newcommand{\opcat}[1]{#1^{\text{op}}}

\newcommand{\autorule}[1]{\relax\ifmmode{\scriptstyle({#1})}\else$(#1)$\fi}

\newcommand{\callcc}{\ensuremath{\mathtt{call/cc}}}
\newcommand{\throw}[2]{\ensuremath{\mathtt{throw}_{#2}^{#1}}}

\let\rq\relax
\let\lq\relax
\DeclareMathSymbol{\lq}{\mathord}{operators}{``}
\DeclareMathSymbol{\rq}{\mathord}{operators}{`'}

\newcommand{\pole}{\bot\!\!\!\bot}
\newcommand{\setpl}{\mathrm{PL}}

\DeclareMathSymbol{\mlq}{\mathord}{operators}{``}
\DeclareMathSymbol{\mrq}{\mathord}{operators}{`'}

\newcommand{\cc}{\mathsf{CC}}

\newcommand{\emem}{e_{\cc}}

\newcommand{\llangle}{\langle\!|}
\newcommand{\rrangle}{|\!\rangle}

\newcommand{\eid}{e_{\mathtt{id}}}

\newcommand{\ecompos}[2]{#1 \mathop{;} #2}

\newcommand{\etrue}{e_{\scriptscriptstyle\top}}

\newcommand{\epair}[2]{\llangle #1, #2 \rrangle}

\newcommand{\efst}{e_{\mathtt{fst}}}

\newcommand{\esnd}{e_{\mathtt{snd}}}

\newcommand{\elambda}[1]{\lambda #1}

\newcommand{\eeval}{e_{\mathtt{eval}}}

\newcommand{\imp}{\supset}
\newcommand{\power}{\mathcal{P}}

\newcommand{\progset}{\mathrm{P}}
\newcommand{\progtypex}[1]{\progset}

%% file: section/introduction.tex
\section{Introduction}
\label{sec:introduction}

Our goal in this paper is to establish the concept of an \emph{effectful topos}, which is an extension of the effective topos that takes into account computational effects,
and to characterize the Lawvere-Tierney topologies and sheaves thereon.
The {effective topos} constitutes the most fundamental framework in realizability theory for analyzing constructive mathematics via computational methods.
While realizability involving computational effects such as monads has been actively studied in recent years,
attempts to lift these developments to topos theory remain insufficient.
In addition to effects, another way to extend computation is through oracles.
As we will discuss later, studying Lawvere-Tierney topologies on effectful toposes
is expected to play a pivotal role in providing a unified connection between realizability relativized by oracles and by effects.

The theory of realizability provides a vital connection between constructive mathematics and computability theory.
The seminal work of Kleene \cite{kleene_interpretation_1945} developed ``effective proofs'' for sentences in Heyting arithmetic using recursive functions,
thereby allowing for a formulation of the Brouwer-Heyting-Kolmogorov (BHK) interpretation \cite{kolmogoroff_zur_1932,heyting_intuitionism_1956} in constructivism.
Moreover, Hyland \cite{hyland_effective_1982} lifted realizability to the framework of topos theory with his discovery of the effective topos constructed from Kleene's model $\kleenek$.
Traditionally, such realizability models are based on partial combinatory algebras (PCAs) \cite{feferman_language_1975, hofstra_partial_2004, van_oosten_partial_2008},
which abstractly represent the concepts of computability and function application.

Recently, realizability involving computational effects has been actively studied.
While standard PCA-based models represent pure computations,
they often lack the flexibility to handle {effects}, such as non-determinism, state, or exceptions.
To bridge this gap, the frameworks of {evidenced frames} \cite{cohen_evidenced_2021} and {monadic combinatory algebras (MCAs)} \cite{cohen_partial_2025} were introduced.
With these frameworks, it has become possible to analyze the relationship between monads and the logic they realize in a unified manner.

However, efforts to translate these approaches into the language of topos theory have been limited.
Cohen et al. \cite{cohen_evidenced_2021} have clarified the relationship between evidenced frames and triposes \cite{hyland_tripos_1980}.
By combining these with the standard topos construction, we can obtain realizability toposes induced from evidenced frames. We call these structures effectful toposes in this paper.
Despite the logical utility of evidenced frames, the properties of the effectful toposes have not yet been fully explored.

In particular, the interplay between these effectful structures and Lawvere-Tierney topologies \cite{lawvere_quantifiers_1970, tierney_sheaf_1972} remains an open area of investigation.
Kihara \cite{kihara_lawvere-tierney_2023} demonstrated that Lawvere-Tierney topologies on the standard effective topos correspond to oracle computations.
Establishing a similar correspondence in the context of effectful toposes would provide a categorical account of the connection between computational effects and oracles.

In this paper, we investigate the foundations of effectful toposes with a focus on their topological structure. Our main contributions are as follows.
\begin{itemize}
    \item Construction \ref{const:EFT} formalizes the ``evidenced frame-to-topos'' construction
        using $\UFam$ \cite{cohen_evidenced_2021} and Pitts' constructions \cite{hyland_tripos_1980,pitts_theory_1981}, named $\eftopos$.
    \item Theorem \ref{thm:main_result} identifies a necessary and sufficient condition for the sheaf topos induced by a topology
        on the effectful topos isomorphic to the effective topos.
\end{itemize}
Finally, Section \ref{sec:conclusion} concludes with a discussion on future challenges towards connecting effectful realizability and oracle realizability.

%% file: section/ef_to_topos.tex
\section{Effectful Topos}
\label{sec:ef_to_topos}
\setcounter{general}{1}

In this section, we demonstrate how a realizability topos can be constructed from an evidenced frame.

Introduced by Cohen et al. \cite{cohen_evidenced_2021}, an evidenced frame is a framework that abstracts the general properties of the BHK interpretation
to provide a unified system of propositions and evidences.
Consequently, this allows us to construct models based on syntactic calculi,
rather than being restricted to specific algebraic structures like PCAs.
We begin with its definition.

\begin{definition}[Evidenced Frames {\cite[Def. III.1]{cohen_evidenced_2021}}]
    \label{dfn:evidenced_frame}
    A triple $(\Phi, E, \cdot \evrel{\cdot} \cdot)$ consisting of a set $\Phi$ of \emph{propositions}, a set $E$ of \emph{evidences},
    and an \emph{evidence relation} $\phi \evrel{e} \psi$ on $\Phi \times E \times \Phi$, is called an \emph{evidenced frame}
    if it is equipped with the constructs of Table \ref{tab:ef_definition}.
    A proposition $\phi$ is said to be \emph{evidenced} by an evidence $e$ if $\top \evrel{e} \phi$ holds.
    In that case, we say that $\phi$ is \emph{evidenceable}.
\end{definition}

\begin{table}[t]
    \centering
    \scalebox{0.9}{
    \input{figures/evidenced_frame_definition.tex}}
    \caption{Evidenced frame constructs \cite[Fig. 2]{cohen_partial_2025} where lowercase Greek letters are propositions and $\Psi \subseteq \Phi$.}
    \label{tab:ef_definition}
\end{table}

Crucially, the framework generalizes the notion of computational truth through the {evidenceability}, i.e., whether a proposition is evidenceable or not.
This concept serves as the counterpart to standard realizability. 

We also use the following abbreviations:
\begin{align*}
    \phi \imp \psi &:= \phi \imp \qty{\psi},\quad \phi \efequiv \psi := (\phi \imp \psi) \land (\psi \imp \phi), \\
    \prod \Psi &:= \top \imp \Psi, \quad \bot := \prod\Phi, \\
    \coprod \Psi &:= \prod\Set{ \prod \Set{\psi \imp \phi \mid \psi \in \Psi } \imp \phi | \phi \in \Phi }.
\end{align*}
Of course, we can intuitively regard $\prod \Psi$ as ``all $\psi \in \Psi$ hold,'' and $\coprod \Psi$ as ``some $\psi \in \Psi$ holds.''
In particular, for a set $X$ and a predicate $\phi \colon X \to \Phi$, we write as
\begin{align*}
    \prod_{x \in X}\phi(x) := \prod \Set{\phi(x) | x \in X}, \quad
    \coprod_{x \in X}\phi(x) := \coprod \Set{\phi(x) | x \in X}.
\end{align*}

\begin{remark}
    We frequently employ the deduction theorem as a working principle:
    assume that $\phi \imp \psi$ is evidenced by $e$. If $\exists e'. \top \evrel{e'} \phi$ also holds, then
    \begin{equation*}
        \top \evrel{e'}
        \phi \evrel{\epair{e_\top}{\eid}} \top \land \phi
        \evrel{\epair{\efst;e}{\esnd}} (\phi \imp \psi) \land \phi
        \evrel{\eeval} \psi.
    \end{equation*}
    Conversely, suppose that $\phi$ is evidenceable. If we obtain an evidence of $\psi$ in an ``effectful'' manner relative to the base evidenced frame, then $\phi \imp \psi$ is also evidenceable.
    While it is difficult to verify whether a meta-level argument meets the requirements, it is sufficient if the argument is at least constructive.
\end{remark}

Although we can derive a topos from an evidenced frame using known methods, the downside is that the structure of the resulting topos lacks clarity.
Given an evidenced frame $\EF$, one can obtain a $\CSet$-tripos $\UFam(\EF)$ by the $\UFam$ construction \cite{cohen_evidenced_2021}.
Furthermore, the well-known tripos-to-topos construction \cite{hyland_tripos_1980, pitts_theory_1981}, often called ``Pitts' construction,'' yields a topos $\mathcal C[\tripP]$ from a $\mathcal C$-tripos $\tripP$.
With these constructions we have the topos $\eftopos(\EF) := \CSet[\UFam(\EF)]$ induced by $\EF$.
However, as Pitts' construction characterizes the components of $\mathcal C[\tripP]$ using {the language of $\tripP$},
their meaning in the context of evidenced frames remains obscure.
Therefore, we formulate $\eftopos(\EF)$ directly in terms of evidenced frames.

In this paper,
we denote the meet, join, implication, the greatest and least elements (if they exist) of a Heyting prealgebra as $\meet, \join, \himp, \htop$ and $\hbot$, respectively.
Let $\CatHp$ be the preorder-enriched category of Heyting prealgebras and structure-preserving maps ordered pointwise.
For a c.c.\ category $\mathcal C$, a pseudofunctor $\tripP \colon \opcat{\mathcal C} \to \CatHp$ is a \emph{$\mathcal C$-tripos}
if it satisfies the following conditions \cite{van_oosten_realizability_2008}.
\begin{enumerate}
    \item For every $f \colon X \to Y$ of $\opcat{\mathcal C}$, the morphism $\tripP(f) \colon \tripP(Y) \to \tripP(X)$ that is also written as $f^*$
        has left/right adjoints $\exists_f/\forall_f$ in the category of preorders.
    \item For every pullback square
        \begin{equation*}
            \begin{tikzcd}
                X \arrow[r, "f"] \arrow[d, "g"'] & Y \arrow[d, "h"] \\
                Z \arrow[r, "k"'] & W
            \end{tikzcd}
        \end{equation*}
        in $\mathcal C$, the compositions $\forall_f \circ g^*$ and $h^* \circ \forall_k$ are isomorphic.
    \item There is a \emph{generic element}: a pair of $\Sigma \in \mathcal C$ and $\sigma \in \tripP(\Sigma)$ with the property that,
        for any $X \in \mathcal C$ and any $\phi \in \tripP(X)$,
        there is a morphism $\tripc{\phi} \colon X \to \Sigma$ making $\phi \simeq \tripc{ \phi }^*(\sigma)$ in $\tripP(X)$.
\end{enumerate}

For an evidenced frame $\EF = (\Phi, E, \cdot \evrel{\cdot}\cdot)$, the $\UFam$ construction gives rise to the $\CSet$-triposes $\UFam(\EF)$ \cite[Thm. V.5]{cohen_evidenced_2021}.
This can be explicitly described as follows.
For a set $X$ and predicates $\phi, \psi \colon X \to \Phi$, if we define the preorder $\le_X$ by
\begin{math}
    \phi \le_X \psi := \exists e \in E. \forall x \in X.\ \phi(x) \evrel{e} \psi(x),
\end{math}
then the pair $(\Phi^X, \le_X)$ forms a Heyting prealgebra.
One determines the mapping $\UFam(\EF) \colon \opcat{\CSet} \to \CatHp$ as
\begin{equation*}
    \UFam(\EF)(X) = \Phi^X, \quad \UFam(\EF)\qty(f \colon X \to Y) = \lambda \phi \in \Phi^Y.\ \phi \circ f.
\end{equation*}
This is a $\CSet$-tripos whose adjoints for $f \colon X \to Y$ and generic element are given by the following:
\begin{equation*}
    \begin{aligned}
        \exists_f &= \lambda\phi \in \Phi^Y. \lambda y \in Y.\ \coprod \Set{\phi(x) | x \in X \land f(x) = y}, \\
        \forall_f &= \lambda\phi \in \Phi^Y. \lambda y \in Y.\ \prod \Set{\phi(x) | x \in X \land f(x) = y}, \\
        \Sigma &= \Phi, \quad \sigma = \id_\Phi, \quad \tripc{ \phi } = \phi \ \text{for any $X$ and $\phi \in \Phi^X$.}
    \end{aligned}
\end{equation*}

Note that $\Phi^X$ is not necessarily complete, but it does have the greatest and least elements. They are obviously given as $\htop = \lambda x. \top$ and $\hbot = \lambda x. \bot$.
Hence, we can write down the structure of the target topos as follows.
\begin{construction}
    \label{const:EFT}
Let $\EF = (\Phi, E, \cdot \evrel{\cdot} \cdot)$ be an evidenced frame.
The structure $\eftopos(\EF)$ consists of the following components.
\end{construction}

\begin{description}
\item[Objects:]
An object of $\eftopos(\EF)$ is a pair $(X, \sim)$ where $X$ is a set and $\sim$ is an \emph{equality predicate} for $X$:
that is, a binary operator $X \times X \to \Phi$ such that the following are evidenceable:
\begin{alignat*}{1}
    \psym &:= \prod_{x \in X} \prod_{y \in X} (x \sim y \imp y \sim x), \\
    \ptrs &:= \prod_{x \in X} \prod_{y \in X} \prod_{z \in X} (x \sim y \land  y \sim z \imp x \sim y).
\end{alignat*}
We adopt the abbreviations $\rsex(x) := x \sim x$, $\rsex(x, y) := \rsex(x) \land \rsex(y)$, and so on.
An element $x \in X$ is said to be an \emph{existent on $X$} if $\rsex(x)$ is evidenceable.

\item[Morphisms:]
Let $(X, \sim)$ and $(Y, \sim)$ be objects. By abuse of notation, we use the same symbol $\sim$ for the relations for both $X$ and $Y$.
A function $F\colon X \times Y \to \Phi$ is a \emph{functional predicate} $(X, \sim) \to (Y, \sim)$
if the following are evidenceable:
\begin{align*}
    \pext(F) &:= \prod_{x \in X}\prod_{x' \in X} \prod_{y \in Y}\prod_{y' \in Y}(x \sim x' \land y \sim y' \land F(x, y) \imp F(x', y')), \\ 
    \psv(F) &:= \prod_{x \in X} \prod_{y \in Y}\prod_{y' \in Y} (\rsex(x, y, y') \land F(x, y) \land F(x, y') \imp y \sim y'), \\ 
    \ptot(F) &:= \prod_{x \in X} \qty\Big(\rsex(x) \imp \coprod_{y \in Y} (\rsex(y) \land F(x, y))). 
\end{align*}
Moreover, functional predicates $F, G \colon (X, \sim) \to (Y, \sim)$ are \emph{equivalent} if
\begin{equation*}
    \peq(F, G) := \prod_{x \in X} \prod_{y \in Y} (\rsex(x, y) \imp (F(x, y) \efequiv G(x, y)))
\end{equation*}
is evidenceable.
A morphism $(X, \sim) \to (Y, \sim)$ is an equivalence class of functional predicates $(X, \sim) \to (Y, \sim)$ under the equivalence above.
Then, a morphism $[F]$ is \emph{represented} by $F$.

\end{description}

We defer the construction of composition and identities to \ref{sec:appendix_EFT} due to space limitations.

For a function $f \colon X \to Y$ of $\CSet$,
the lifting predicate $\lift{f} \colon X \times Y \to \Phi$ defined by $\lift{f}(x, y) := f(x) \sim y$
is functional if and only if the following is evidenceable:
\begin{equation*}
    \pfext(f) := \prod_{x \in X} \prod_{x' \in X} (x \sim x' \imp f(x) \sim f(x')).
\end{equation*}
In that case, one says that $f$ is \emph{extensional} and that it \emph{represents} the morphism $\eqclass{ \lift{f} } \colon (X, \sim) \to (Y, \sim)$.
Conversely, every morphism $(X, \sim) \to (Y, \sim)$ is represented by some function $X \to Y$ (with the Axiom of Choice).
Hereafter, we often denote a function $f$ and the morphism $\eqclass{\lift{f}}$ by the same $f$.

It is easy to show that $\eftopos(\EF)$ forms a category.
To verify that this is moreover a topos, we need to specify further structure:
namely, that $\eftopos(\EF)$ has a terminal object, finite products, equalizers, and a subobject classifier.

\begin{description}
\item[Terminal object:]
    Letting $\singleton := \qty{*}$ be a singleton set,
    the terminal object is given by $(\singleton, \htop)$.

\item[Subobject classifier:]
    Since $\top$ is also seen as the function $* \mapsto \top$,
    the lifting $\lift{\top}(*, \phi) = \phi \efequiv \top$ represents the morphism $\true \colon (\singleton, \htop) \to (\Phi, \efequivop)$.
    Then, $(\Phi, \efequivop)$ with $\true$ is a subobject classifier.
    That is, for any monomorphism $m \colon (S, \sim) \hookrightarrow (X, \sim)$, there is a unique \emph{classifying morphism} $\chi_m \colon (X, \sim) \to (\Phi, \efequivop)$ making
    $m$ be the pullback of $\true$.
\end{description}
See \ref{sec:appendix_EFT} for details on finite products and equalizers.

\begin{proposition}
    $\eftopos(\EF)$ is an elementary topos.
\end{proposition}

The author proposes the term \emph{effectful topos} for the realizability topos
$\eftopos(\EF)$ constructed from an evidenced frame $\EF$.

\begin{remark}
    \label{remark:subobject}
    Again, a monomorphism $m \colon (S, \sim) \hookrightarrow (X, \sim)$ one-to-one corresponds to the classifying morphism $\chi_m$. 
    Recall that a subobject is an isomorphism class of monomorphisms. The subobject $(S, \sim)$ has the canonical representation $i_m \colon (X, \sim_m) \hookrightarrow (X, \sim)$
    where $x \sim_m x' := x \sim x' \land \chi_m(x)$ and $i_m$ is the inclusion morphism represented by $\id_X$, as the diagram below:
    \begin{equation*}
        \begin{tikzcd}[row sep=large, column sep=large]
            (X, \sim_m) \arrow[rd, dashed, xshift=0.5ex, yshift=0.9ex, "h"] \arrow[rdd, bend right, hook, "i_m"] \arrow[rrd, bend left, "!"']\\
        & (S, \sim) \arrow[r, "!"] \arrow[d, hook, "m"'] \arrow[lu, xshift=-0.5ex, yshift=-0.9ex, "m'", "\cong" sloped] & (\singleton, \htop) \arrow[d, hook, "\true"] \\
        & (X, \sim) \arrow[r, dashed, "\chi_m"'] & (\Phi, \efequivop)\rlap{.}
        \end{tikzcd}
    \end{equation*}
\end{remark}

%% file: figures/evidenced_frame_definition.tex
    \begin{tabular}{l|l|l|l}
        \hline
        Axiom &
        Logical Const. &
        Program Const. &
        Evidence Relation \\ \hline \hline
        \emph{Reflexivity} &
                           &
        $\eid \in E$ &
        $\phi \evrel{\eid} \phi$ \\ \hline
        \emph{Transitivity} &
                            &
        $; \colon E^2 \to E$ &
        $\phi_1 \evrel{e_1} \phi_2$ and $\phi_2 \evrel{e_2} \phi_3 \Rightarrow \phi_1 \evrel{\ecompos{e_1}{e_2}} \phi_3$ \\ \hline
        \emph{Top} &
        $\top \in \Phi$ &
        $\etrue \in E$ &
        $\phi \evrel{\etrue} \top$ \\ \hline
        \emph{Conjunction} &
        $\land \colon \Phi^2 \to \Phi$ &
        \begin{tabular}[c]{@{}l@{}}$\epair{\cdot}{\cdot} \colon E^2 \to E$\\ $\efst, \esnd \in E$\end{tabular} &
        \begin{tabular}[c]{@{}l@{}}$\phi \evrel{e_1} \psi_1$ and $\phi \evrel{e_2} \psi_2 \Rightarrow \phi \evrel{\epair{e_1}{e_2}} \psi_1 \land \psi_2$\\ $\phi_1 \land \phi_2 \evrel{\efst} \phi_1$ and $\phi_1 \land \phi_2 \evrel{\esnd} \phi_2$\end{tabular} \\ \hline
        \emph{Univ. Implication} &
        $\imp \colon \Phi \times \power(\Phi) \to \Phi$ &
        \begin{tabular}[c]{@{}l@{}}$\lambda \colon E \to E$\\ $\eeval \in E$\end{tabular} &
        \begin{tabular}[c]{@{}l@{}}$(\forall \psi \in \Psi.\phi_1 \land \phi_2 \evrel{e} \psi) \Rightarrow \phi_1 \evrel{\elambda e} \phi_2 \imp \Psi$\\ $\phantom{(}\forall \psi \in \Psi.(\phi \imp \Psi) \land \phi \evrel{\eeval} \psi$\end{tabular} \\ \hline
    \end{tabular}

%% file: section/mca_k1.tex
\section{Effective Topos as Effectful Topos}
\label{sec:mca_k1}
\setcounter{general}{1}

Our objective in this section is to introduce monadic combinatory algebras (MCAs) and demonstrate that the effective topos can be viewed as a specific example of effectful toposes through this structure.
Partial combinatory algebras (PCAs) are structures that algebraically represent the concept of computability formulated by Turing machines and $\lambda$-calculi.
Kleene's model $\kleenek$ is a well-known instance of a PCA.
In the context of the theory of computation, the effective topos constructed from $\kleenek$ is a remarkably important example of a topos.
MCAs are extensions of PCAs that incorporate effectful computation, and $\kleenek$ can also be constructed as an MCA.

We briefly recall the construction of evidenced frames from MCAs, referring to Cohen et al. \cite{cohen_partial_2025} for full details.
Here, we consider only monads over $\CSet$. For a monad $M$, let $\eta \colon \id \Rightarrow M$ and $\mu \colon MM \Rightarrow M$ denote its return and bind, respectively.
An \emph{MCA over a monad $M$} is a pair $(\algebraA, \cdot)$ consisting of a set $\algebraA$ of \emph{codes}
and a \emph{Kleisli application} $\cdot \colon \algebraA \times \algebraA \to M(\algebraA)$ that is combinatory complete \cite[Def. 5]{cohen_partial_2025}.
Computation is represented by \emph{expressions} constructed via codes, variables and a formal binary operator $\bullet$.
We write $E(\algebraA)$ for the set of expressions of $(\algebraA, \cdot)$.
Expressions without variables are interpreted by the \emph{evaluation} $\nu$ given by
\begin{align*}
    \nu(c) := \eta_{\algebraA}(c), \quad
    \nu(e_1 \bullet e_2) := \mu_{\algebraA} \circ M(\lambda c_1. \mu_\algebraA \circ M(\lambda c_2. c_1 \cdot c_2)(\nu(e_2)))(\nu(e_1))
\end{align*}
for a code $c$ and expressions $e_1, e_2$.
\begin{example}
    \label{example:pca_as_mca}
    The following is one of the most fundamental and notable examples of an MCA.
    Let a monad $M_1$ be
    \begin{align*}
        M_1(A) := \Set{ S \subseteq A | \left|S\right| \le 1}, \quad
        \eta_A(x) := \qty{x}, \quad \mu_A(m) := \bigcup_{X \in m} X.
    \end{align*}
    Moreover, take the natural numbers $\Natural$ as the codes,
    and let the application be $a \cdot b := \qty{\varphi_a(b)}$ where $\varphi_i \colon \Natural \rightharpoonup \Natural$ denotes the $i$-th partial computable function.
    Note that $a \cdot b$ is a singleton set if $\varphi_a(b)$ is defined and empty otherwise.
    This $(\Natural, \cdot)$ is called \emph{Kleene's first model} $\kleenek$.
\end{example}

An \emph{$M$-modality} is a natural transformation $\modality \colon M \Rightarrow ((- \to \Omega) \to \Omega)$ with a complete Heyting prealgebra $\Omega$
satisfying the consistency axioms \cite[Def. 13]{cohen_partial_2025}.
We also denote $\modalangle{x}{m}{\phi(x)} := \modality_\algebraA(m)(\phi)$ for $m \in M(\algebraA)$ and $\phi \colon \algebraA \to \Omega$.
Furthermore, we fix a \emph{$\modality$-separator} $\separator \subseteq \algebraA$ that satisfies the progress property,
which intuitively guarantees that computations in $\separator$ yield values.
Finally, the tuple $(\algebraA, \Omega, \modality, \separator)$ constitutes a \emph{monadic core}.
Every monadic core induces an evidenced frame $(\Omega^\algebraA, \separator, \cdot \evrel{\cdot} \cdot)$, where the evidence relation is defined by
\begin{equation*}
    \phi \evrel{e} \psi := \forall c \in \algebraA.\, \phi(c) \le \modalangle{r}{e \cdot c}{\psi(r)}
\end{equation*}
for predicates $\phi, \psi \colon \algebraA \to \Omega$ and a code $e \in \separator$ \cite[Thm. 16]{cohen_partial_2025}.

Following Example \ref{example:pca_as_mca}, we see the two ways of creating toposes from Kleene's $\kleenek$, and their coincidence.
First, let us describe the traditional construction.
The set $\power(\kexps)^X$ forms a Heyting prealgebra by the order
\begin{equation*}
    \phi \le \psi := \exists a \in \kexps, \forall x \in X, \forall b \in \phi(x).\ \nu(a \bullet b) \neq \emptyset \land \nu(a \bullet b) \subseteq \psi(x)
\end{equation*}
for a set $X$ and predicates $\phi, \psi \colon X \to \power(\kexps)$.
This gives rise to a pseudofunctor $\tripET(-) := \power(\kexps)^{(-)} \colon \opcat{\CSet} \to \CatHp$ by letting the mapping on functions be
\begin{equation*}
    \begin{array}{cccc}
        \tripET(f) : & \power(\kexps)^{Y} & \to & \power(\kexps)^{X} \\
                     &  \phi & \mapsto & \phi \circ f
    \end{array}
\end{equation*}
for each $f \colon X \to Y$. Since $\tripET$ is indeed a $\CSet$-tripos, 
we can get the \emph{effective topos} $\CSet[\tripET]$ by Pitts' construction.

On the other hand, we have another method for constructing a topos from $\kleenek$.
Let $\boolhyt := \qty{\hbot, \htop}$ be the complete Heyting algebra of Booleans; that is, there is a unique non-trivial relation $\hbot \le \htop$ on $\boolhyt$.
Note that 
any $m \in M_1(A)$ is a singleton subset of $A$ or the emptyset for each set $A$ by the definition of $M_1$.
Therefore, the $M_1$-modality $\modality_\kleenek \colon M_1 \Rightarrow ((-) \to \boolhyt) \to \boolhyt$ can be defined by
\begin{equation*}
    \makleene{x}{m}{\phi(x)} := \left\{
        \begin{array}{cl}
            \phi(a) & (m = \qty{a}) \\
            \hbot & (m = \emptyset)
        \end{array}
        \right.
\end{equation*}
for every $\phi \colon A \to \boolhyt$ and $m \in M_1(A)$.
Finally, the whole code $\Natural$ is a separator.
The evidenced frame $\EFkleene := (\effomega, \Natural, \cdot \evrel{\cdot} \cdot)$ is constructed from the monadic core $(\kleenek, \boolhyt, \modality_\kleenek, \Natural)$ \cite[Ex. 19]{cohen_partial_2025}.
Then, we have the tripos $\UFam(\EFkleene)$ and the topos $\eftopos(\EFkleene)$.

We observe a clear structural correspondence between $\tripET$ and $\UFam(\EFkleene)$.
Fundamentally, this arises from the canonical isomorphism between the powerset $\power(E(\Natural))$ and the exponential $\effomega$ via the characteristic function.
Specifically, for each set $X$, we define the function $\Char_X \colon \power(E(\Natural))^X \to (\effomega)^X$ by
\begin{equation*}
    \Char_X(\phi) := \lambda x \in X.\lambda a \in \Natural.\ \left\{
        \begin{array}{cl}
            \htop & (a \in    \nu(\phi(x))) \\
            \hbot & (a \notin \nu(\phi(x))),
        \end{array}
        \right.
\end{equation*}
where $\nu(K) := \bigcup_{e \in K}\nu(e)$ for $K \in \power(E(\Natural))$.
It is straightforward to verify that the collection $\Char_{(-)}$ constitutes a natural isomorphism $\tripET \Rightarrow \UFam(\EFkleene)$.
This establishes a one-to-one functor mapping each object $(X, \sim)$ of $\CSet[\tripET]$ to the object $(X, \Char_{X \times X}(\sim))$ of $\eftopos(\EFkleene)$,
and each morphism $\eqclass{F} \colon (X, \sim) \to (Y, \sim)$ to the morphism $\eqclass{\Char_{X \times Y}(F)} \colon (X, \Char_{X \times X}(\sim)) \to (Y, \Char_{Y \times Y}(\sim))$, respectively.

\begin{proposition}
    \label{prop:eff_efk_isomorphism}
    The characteristic functor $\Char \colon \CSet[\tripET] \to \eftopos(\EFkleene)$ gives an isomorphism.
\end{proposition}

We can identify $\CSet[\tripET]$ with $\eftopos(\EFkleene)$ by Proposition \ref{prop:eff_efk_isomorphism}.
Then, we also refer to the latter as the effective topos in this paper.

%% file: section/topology.tex
\section{Lawvere-Tierney Topologies and Sheaves}
\setcounter{general}{1}
\label{sec:lt_topology}

\begin{sloppypar}
As the final part of our study, we examine {Lawvere-Tierney topologies} on $\toposEff$ and their {sheaves}.
These topologies, also called {local operators} or {nuclei} \cite{van_den_berg_kuroda-style_2019}, are morphisms representing properties on a topos;
given a topos $\toposE$ and a topology $j$, one obtains a new topos $\Sheaf_j(\toposE)$ referred to as a $j$-sheaf topos.
\end{sloppypar}

Of particular importance is the case of the effective topos $\CSet[\tripET]$.
As noted by Kihara \cite{kihara_lawvere-tierney_2023}, for each topology, we can consider an oracle computation and its corresponding intermediate logic.
Then, we can directly utilize these results
by leveraging the isomorphism in the previous section.
That is, it suffices to transfer concepts from the effective topos via the characteristic functor to define their counterparts in the effectful topos $\toposEff$.

Let us define the Lawvere-Tierney topology.
Although this concept was originally introduced by Lawvere \cite{lawvere_quantifiers_1970} and Tierney \cite{tierney_sheaf_1972},
here we adopt an approach that is more suited for analyzing the effective topos \cite{kihara_lawvere-tierney_2023,hyland_effective_1982,lee_subtoposes_2011,lee_basic_2013}.

\begin{definition}[Lawvere-Tierney Topology]
    \label{dfn:LT_topology}
    A \emph{Lawvere-Tierney topology} on $\toposEff$ is a function $j \colon \effomega \to \effomega$ such that the following are evidenceable:
    \begin{align*}
        \pinc &:= \prod_{\phi \in \Phi}(\phi \imp j(\phi)), \quad \pidm := \prod_{\phi \in \Phi}(j(j(\phi)) \imp j(\phi)), \\
        \pprs &:= \prod_{\phi \in \Phi} \prod_{\psi \in \Phi}(j(\phi \land \psi) \efequiv j(\phi) \land j(\psi)).
    \end{align*}
\end{definition}

Roughly speaking, a Lawvere-Tierney topology $j$ introduces a ``world'' where truth conditions are weakened.
The condition $\pinc$ implies that even if a proposition $\phi$ is not evidenceable, $j(\phi)$ might be. Let us informally refer to this state as ``holding in the $j$-world.''
If $S$ is a subobject of $X$, the existents of $S$ may be strictly smaller than those of $X$.
However, there are cases where they coincide within the $j$-world; such subobjects are called dense as defined below.

In what follows, we simply write $X$ for an object $(X, \sim)$ with an arbitrary equality $\sim$, provided no confusion arises.
Specifically, $\effomega$ and $\singleton$ denote $(\effomega, \efequiv)$ and $(\singleton, \htop)$, respectively.

\begin{definition}[Density]
    \label{dfn:topos_closure}
    Let $j \colon \effomega \to \effomega$ be an extensional function.
    Given a subobject $m \colon S \hookrightarrow X$,
    its \emph{$j$-closure} $\overline{m} \colon \overline{S} \hookrightarrow X$ is the subobject given by $\chi_{\overline m} := j \circ \chi_m$.%
    \footnote{Note that $\overline S$ is uniquely determined as it is the pullback of $\chi_{\overline m}$ and $\true$.}
    In particular, $m$ is said to be \emph{$j$-dense} (or simply \emph{dense}) if $\overline S \cong X$ holds.
\end{definition}

Since all Lawvere-Tierney topologies are extensional, one can consider the density for them.

\begin{definition}[Sheaf]
    \label{dfn:sheaves}
    Let $j$ be a Lawvere-Tierney topology on $\toposEff$.
    An object $A$ is said to be \emph{$j$-separated} if, for every $j$-dense subobject $m \colon S \hookrightarrow X$, the following mapping $m_A^*$ is injective:
    \begin{equation*}
        \begin{array}{rccc}
            m_A^* \colon & \toposEff(X, A) & \to & \toposEff(S, A) \\
                         &  f & \mapsto & f \circ m.
        \end{array}
    \end{equation*}
    Furthermore, $A$ is a \emph{$j$-sheaf} if $m_A^*$ is bijective.
\end{definition}

We write $\Sheaf_j(\toposEff)$ for the full subcategory consisting of all the $j$-sheaves on $\toposEff$,
which forms a topos \cite{mac_lane_topoi_1994}. Thus, we call $\Sheaf_j(\toposEff)$ the \emph{$j$-sheaf topos} of $\toposEff$.

Our main result characterizes the sheaves within an effectful topos.
First, we introduce another term. Let $(A, \sim_A)$ be an object.
For a given $j$-dense subobject $m \colon S \hookrightarrow X$ and a given morphism $g \colon S \to A$,
the \emph{$j$-singleton for $m$ and $g$} is the extensional function $\psi_{-} \colon X \to A \to \effomega$ defined as
\begin{equation*}
    \psi_x(a) := j\qty(\coprod_{s \in S} (\rsex(s) \land m(s) \sim x \land g(s) \sim_A a)).
\end{equation*}

\begin{theorem}
    \label{thm:main_result}
    Let $j$ be a Lawvere-Tierney topology on $\toposEff$ and $(A, \sim_A)$ be an object.
    \begin{enumerate}
        \item $A$ is $j$-separated if and only if the proposition
            \begin{equation*}
                \psep(A) := \prod_{a \in A} \prod_{b \in A} (\rsex(a, b) \land j(a \sim_A b) \imp a \sim_A b)
            \end{equation*}
            is evidenceable.
            \label{enumerate:main_result_1}
        \item $A$ is a $j$-sheaf if and only if it is $j$-separated and
            \begin{equation*}
                \pdsc(A) := \prod_{x \in X}\coprod_{a \in A} \qty(\rsex(a) \land \prod_{b \in A}(\rsex(b) \imp (\psi_x(b) \efequiv j(a \sim_A b)) ))
            \end{equation*}
            is evidenceable for any $g$ and $m$.
            \label{enumerate:main_result_2}
    \end{enumerate}
\end{theorem}

We provide a proof of this theorem in \ref{sec:appendix_proof_theorem}.

The concept of $j$-singletons serves as a tool to show the surjectivity of $m_A^*$, that is, to extend any morphism $g \colon S \to A$ to a new morphism $f \colon X \to A$.
Intuitively, this can be explained as follows: since $S$ is a subobject, there may be existents $x \in X$ that do not exist in $S$.
However, due to the density of $S$, an existent $s$ satisfying $x \sim m(s)$ always exists in the $j$-world.
The singleton $\psi_x(a)$ describes the condition ``for an existent $s$ equal to $x$, $a$ is equal to $g(s)$'' within the $j$-world.
The condition $\pdsc(A)$ guarantees that a unique $a_x$ satisfying this description exists within $A$.
The morphism $f$ mapping $x$ to this $a_x$ is precisely the desired morphism satisfying $f \circ m = g$.

\subsection{Effectful Topos via CPSCA vs. $\lnot\lnot$-Sheaf Topos}
\label{ssec:cpsca_vs_dnt_topos}

To demonstrate the applicability of Theorem \ref{thm:main_result}, we show that two specific toposes coincide:
one is a sheaf topos induced by a topology, and the other is an effectful topos constructed from an MCA. Both serve as models of classical logic.

We begin with the former. The topos $\Sheaf_{\lnot\lnot}(\toposEff)$ yielded by the topology $\lnot\lnot$ is a model of classical realizability via the double negation translation \cite{godel_intuitionistic_2001,gentzen_uber_1974,van_den_berg_kuroda-style_2019}.
\begin{definition}
    \label{dfn:dne-topology}
    We define the function $\lnot \colon \effomega \to \effomega$ by $\lnot\phi := \phi \imp \bot$.
    The \emph{double negation topology} $\lnot\lnot$ is given by $\lnot\lnot \phi := \lnot(\lnot \phi)$.
\end{definition}
One can check that $\lnot\lnot$ is indeed a Lawvere-Tierney topology on $\toposEff$.
It is well-known that $\Sheaf_\dnt(\toposEff)$ is equivalent to $\CSet$ \cite{hyland_effective_1982}; this topos represents the world of set-theoretic mathematics.

In parallel, there is an MCA with capabilities similar to the above.
Krivine \cite{krivine_typed_2001,krivine_realizability_2009} pioneered a computational approach where double negation elimination is realized by capturing and restoring continuations on a stack.
The {continuation-passing-style combinatory algebra (CPSCA)} \cite{cohen_partial_2025} formalizes this mechanism within the MCA framework,
and induces an evidenced frame analogous to type-theoretic constructions \cite{gardelle_cps_2023}.

\begin{construction}
    \label{example:topos_cpsca}
    A \emph{CPSCA} is an instance of MCA whose monad is given by
    \begin{align*}
        M_\cps(A) := (A \to R) \to R, \quad
        \eta_A(a) := \lambda k.k(a), \quad \mu_A(m) := \lambda k. m(\lambda g.g(k)),
    \end{align*}
    where $R$ is an arbitrary set. Roughly, $R$ represents the ultimate results of the whole computation \cite[\S 3.2.3]{cohen_partial_2025}.
    This algebra possesses codes $\emem$ and $\ethrow{k}$ that save and replace continuations (the so-called \callcc\, and \throw{}{}),
    making the following hold for all $c, c' \in \Natural$ and $k, k' \colon \Natural \to \Natural$:
    \begin{equation*}
        (\emem \cdot c)(k) = (c \cdot \ethrow{k})(k), \quad (\ethrow{k} \cdot c')(k') = k(c').
    \end{equation*}
    Furthermore, the pole $\pole$ is a subset of $R$, and we write $f \orthog a$ for a function $f \colon A \to R$ and an element $a \in A$ when $f(a) \in \pole$.
    Intuitively, $\pole$ is the set of valid computation results, and $f \orthog a$ signifies that the computation $f$ succeeds on the input $a$.
    Again, $\meet$ and $\himp$ denote the meet and implication of the Heyting prealgebra $\boolhyt$.
    We define the $M_\cps$-modality $\modality_\cps$ by
    \begin{align*}
        \macps{x}{m}{\phi(x)} &:= \bigmeet_{k \in R^A} \qty(\qty(\bigmeet_{a \in A}(\phi(a) \himp k \orthog a)) \himp m \orthog k).
    \end{align*}
    Aligning with our setting for $\toposEff$, we take $\algebraA := \Natural, R := \Natural, a \cdot b := \lambda k.k(\varphi_a(b))$, and $\pole := \qty{0}$. Let us denote this algebra by $\algebracps$. 
    Finally, for the separator, we take the set $\setpl \subseteq \Natural$ of ``proof-like'' codes. This can be explicitly defined as the combinatory complete closure generated by $\codeS, \codeK$, and $\emem$,
    where $\codeS$ and $\codeK$ are the codes corresponding to the S and K combinators, respectively.
    The tuple $(\algebracps, \boolhyt, \modality_\cps, \setpl)$ forms a monadic core and induces the evidenced frame $\EF_\cps := (\effomega, \setpl, \cdot \evrel{\cdot} \cdot)$ \cite[Ex. 22]{cohen_partial_2025}.
    Thus, we obtain the effectful topos $\toposE_\cps := \eftopos(\EF_\cps)$.
\end{construction}

We now turn to our central question: are $\Sheaf_{\dnt}(\toposEff)$ and $\toposE_\cps$ equivalent? In fact, these toposes are moreover isomorphic.
To establish this result, we require the following key lemmata whose proofs are provided in \ref{sec:appendix_proof_cpsca_dnt}.

\begin{lemma}
    \label{lem:k1_cps_extension}
    If $\phi \in \effomega$ is evidenceable in $\EF_\kleenek$, then it is also evidenceable in $\EF_\cps$.
\end{lemma}

\begin{lemma}
    \label{lem:dne}
    For any $\phi \in \effomega$, $\lnot\lnot \phi \imp \phi$ is evidenceable in $\EF_\cps$.
\end{lemma}

In particular, Lemma \ref{lem:dne} that entails the evidenceability of the double negation elimination plays a central role in the following theorem.

\begin{theorem}
    \label{thm:cps_dnt_isomorphism}
    There exists an isomorphism $\Sheaf_\dnt(\toposEff) \cong \toposE_\cps$.
\end{theorem}

\begin{proof}
    If $(A, \sim_A)$ is an object of $\Sheaf_{\dnt}(\toposEff)$, then it is also one of $\toposE_\cps$ because $\sim_A$ makes $\psym$ and $\ptrs$ be evidenceable in $\EF_\cps$ by Lemma \ref{lem:k1_cps_extension}.

    Conversely, suppose that $(A, \sim_A)$ is an object of $\toposE_\cps$.
    Note that $\toposE_\cps$ is a subtopos of $\toposEff$.
    First, $\psep(A)$ is evidenceable by Lemma \ref{lem:dne}.
    For the evidenceability of $\pdsc(A)$, let $m \colon S \hookrightarrow X$ be a $\dnt$-dense subobject and $g \colon S \to A$ be a morphism.
    Since we have $\pden(m) \efequiv \prod_x (\rsex(x) \imp \chi_m(x))$, $S$ is isomorphic to $X$.
    By taking the existent $g \circ m^{-1}(x) \in A$ for each $x \in X$, we have an evidence of $\pdsc(A)$. Hence, $B$ is a $\dnt$-sheaf and it is an object of $\Sheaf_\dnt(\toposEff)$.
    The case for morphisms is similar.
    \qed
\end{proof}

%% file: section/conclusion.tex
\section{Conclusion and Future Work}
\label{sec:conclusion}

\newcommand{\multicolon}{\mathbin{\colon\!\!\!\subseteq}}

This paper formulated the internal structure of effectful toposes. These are realizability toposes extended via evidenced frames to accommodate effectful computation.
Moreover, we demonstrated that the standard effective topos can be regarded as a specific instance of this framework.
We also explicitly characterized the conditions satisfied by the sheaves induced by Lawvere-Tierney topologies on effectful toposes.

We anticipate that further analysis will reveal a deeper correspondence between computational effects and topologies like Section \ref{ssec:cpsca_vs_dnt_topos}.
For instance, let $\toposE_M$ denote the effectful topos constructed from an MCA over a monad $M$. Then, one might identify the topology $j$ that yields the equivalence $\toposE_M \simeq \Sheaf_j(\toposEff)$, and vice versa.

Furthermore, by formulating an oracle as a \emph{bilayer function} $g \multicolon \Natural \times \Lambda \rightrightarrows \Natural$,
we can derive a corresponding topology $g^{\Game\to} \colon \boolhyt^{\kleenek} \to \boolhyt^{\kleenek}$ on $\toposEff$ \cite{kihara_lawvere-tierney_2023}.
Clarifying the conditions on $g$ and $M$ making $\Sheaf_{g^{\Game\to}}(\toposEff) \simeq \toposE_M$ would establish a unified connection between oracle realizability and effectful realizability.
Naturally, they provide models for intuitionistic logic, classical logic, and many intermediate logics that lie between them.

While our discussion here has focused exclusively on MCAs, the analysis of effectful toposes induced by syntactic systems 
(e.g., $\EffHOL$ \cite{cohen_syntactic_2025}, an extension of Girard's $F_\omega$ \cite{girard_interpretation_1972}) remains an important subject for future work.
In recent years, the formalization of ordinary mathematics using proof assistants such as Lean \cite{noauthor_lean_nodate} and Rocq \cite{noauthor_welcome_2025} has been gaining significant momentum.
Introducing ``effectful proofs'' into these systems could provide simpler or novel methods for proving theorems within intermediate logics.

%% file: section/appendix.tex
\section{$\eftopos$ Construction}
\label{sec:appendix_EFT}

\begin{description}
\item[Composition:]
Let $F \colon (X, \sim) \to (Y, \sim)$ and $G \colon (Y, \sim) \to (Z, \sim)$ be functional predicates.
The composition $G \circ F \colon (X, \sim) \to (Z, \sim)$ is given by
\begin{math}
    G \circ F(x, z) := \coprod_{y \in Y} (\rsex(y) \land F(x, y) \land G(y, z)).
\end{math}
We define the composition of morphisms $\eqclass{F}$ and $\eqclass{G}$ by $\eqclass{G} \circ \eqclass{F} := \eqclass{G \circ F}$;
it is easy to verify that this is well-defined and associative.
If functions $f$ and $g$ represent $\eqclass{F}$ and $\eqclass{G}$, respectively,
then $g \circ f$ represents $\eqclass{G} \circ \eqclass{F}$.

\item[Identity:]
The identity of an object $(X, \sim)$ is the morphism represented by $\id_X \colon X \to X$.
\item[Finite products:]
    The product $(X, \sim) \times (Y, \sim)$ of objects $(X, \sim)$ and $(Y, \sim)$ is defined by
    $(X \times Y, \sim)$ where $(x, y) \sim (x', y') := x \sim x' \land y \sim y'$, and their projections are straightforwardly defined.
    Given two functional predicates $F \colon (W, \sim) \to (X, \sim)$ and $G \colon (W, \sim) \to (Y, \sim)$,
    one can define $\encode{F, G}$ on $W \times (X \times Y) \to \Phi$ by
    \begin{math}
        \encode{F, G}(w, (x, y)) := F(w, x) \land G(w, y).
    \end{math}
    Then, $\encode{F, G}$ is a functional predicate $(W, \sim) \to (X \times Y, \sim)$, and this represents a morphism $\encode{[F], [G]} \colon (W, \sim) \to (X \times Y, \sim)$.

\item[Equalizers:] Let $[F], [G] \colon (X, \sim) \to (Y, \sim)$ be morphisms. Their equalizer is the object $(X, \approx)$ with the morphism $i \colon (X, \approx) \to (X, \sim)$, where
    \begin{math}
        x \approx x' := x \sim x' \land \coprod_{y \in Y} (F(x, y) \land G(x, y))
    \end{math}
    and $i$ is represented by $\id_X$.
\end{description}

\section{Proof of Theorem \ref{thm:main_result}}
\label{sec:appendix_proof_theorem}

This appendix provides a proof of Theorem \ref{thm:main_result}. We begin with the following lemmata.

\begin{lemma}
    \label{prop:morphisms_lefteq_condition}
    Consider an arbitrary effectful topos.
    Let $m \colon (S, \sim) \hookrightarrow (X, \sim)$ be a subobject and $f, g \colon (X, \sim) \to (Y, \sim)$ be morphisms.
    Then, $f \circ m = g \circ m$ holds if and only if
    \begin{math}
        \label{eq:proposition_lefteq}
        \pleq(f, g) := \prod_{x \in X}(x \sim_m x \imp f(x) \sim g(x))
    \end{math}
    is evidenceable.
\end{lemma}

\begin{proof}
    Assume $f \circ m = g \circ m$. Note that this equals $f \circ i_m = g \circ i_m$. 
    Since $i_m$ is represented by $\id_X$, the compositions $f \circ i_m$ and $g \circ i_m$ are respectively represented by extensional functions $f, g \colon X \to Y$.
    That is,
    \begin{equation*}
        \peq(\lift{f \circ i_m}, \lift{g \circ i_m}) = \prod_{x \in X} \prod_{y \in Y} (x \sim_m x \land y \sim y \imp (f(x) \sim y \efequiv g(x) \sim y))
    \end{equation*}
    is evidenceable.
    As $y$ is arbitrary, one obtains an evidence of $\pleq(f, g)$ by taking $y = f(x)$.
    We can similarly derive $\peq(\lift{f \circ i_m}, \lift{g \circ i_m})$ from $\pleq(f, g)$.
    \qed
\end{proof}

\begin{lemma}
    \label{lem:appendix_1}
    Consider the effective topos.
    For an extensional function $j$, a subobject $m \colon S \hookrightarrow X$ is $j$-dense if and only if
    \begin{math}
        \pden(m) := \prod_{x} (\rsex(x) \imp j(\chi_m(x)))
    \end{math}
    is evidenceable.
\end{lemma}

\begin{proof}
    One has an isomorphism $(\overline S, \sim) \cong (X, \sim_{jm})$ where $x \sim_{jm} x' := (x \sim x') \land j(\chi_m(x))$
    and the inclusion $i_{jm} \colon (X, \sim_{jm}) \hookrightarrow (X, \sim)$ by the same argument as in Remark \ref{remark:subobject}.
    If $m$ is dense, then $i_{jm}$ must be isomorphic, that is, the identity function $\id_X$ also represents the isomorphism $(X, \sim) \to (X, \sim_{jm})$.
    This immediately implies that $\pden(m)$ is evidenceable and vice versa. \qed
\end{proof}

We first prove the part (\ref{enumerate:main_result_1}) of the theorem. Again, $j$ is a Lawvere-Tierney topology and $(A, \sim_A)$ is an object of $\toposEff$.

\begin{proof}[Theorem \ref{thm:main_result} (\ref{enumerate:main_result_1})]
    Assume that $A$ is $j$-separated and consider the object $(A^2, \sim_{jA})$ where $(a, b) \sim_{jA} (a', b') := a \sim_A a' \land b \sim_A b' \land j(a \sim_A b)$.
    Let $m \colon S \hookrightarrow A^2$ be the subobject defined as $\chi_m(a, b) := a \sim_A b$.
    Note that $m$ is dense because $(a, b) \sim_{jA} (a, b)$ implies $j(a \sim_A b) = j(\chi_m(a, b))$ by definition.
    For all $f, g \colon (A^2, \sim_{jA}) \to (A, \sim_A)$, if $f \circ m = g \circ m$ holds, we obtain $f = g$ by assumption.
    In particular, we take $f := \pi_1$ and $g := \pi_2$ where $\pi_1$ and $\pi_2$ are the projections. Then, $\pleq(f, g)$ in Lemma \ref{prop:morphisms_lefteq_condition} is evidenceable.
    In this setting, $f \circ m = g \circ m$ holds and so does $f = g$. Therefore, we get an evidence of
    \begin{align*}
        \peq({\lift{f}}, {\lift{g}}) &= \prod_{(a, b) \in A^2} \prod_{a' \in A} (\rsex((a, b), a') \imp (\pi_1(a, b) \sim_A a' \efequiv \pi_2(a, b) \sim_A a')) \\
        &= \prod_{(a, b) \in A^2} \prod_{a' \in A} (\rsex(a, b, a') \land j(a \sim_A b) \imp (a \sim_A a' \efequiv b \sim_A a')) \\
          &\efequiv \prod_{a \in A} \prod_{b \in A} (\rsex(a, b) \land j(a \sim_A b) \imp a \sim_A b) = \psep(A).
    \end{align*}
    For the converse, suppose that $\psep(A)$ is evidenceable and
    let $m \colon S \hookrightarrow X$ be a dense subobject. Then, $\pden(m)$ is evidenceable by Lemma \ref{lem:appendix_1}. 
    If morphisms $f, g \colon X \to A$ satisfy $m_A^*(f) = m_A^*(g)$, i.e., $f \circ i_m = g \circ i_m$, then
    \begin{align*}
        \pleq(f, g) &= \prod_{x \in X}(x \sim_m x \imp f(x) \sim_A g(x)) = \prod_{x \in X}(\rsex(x) \land \chi_m(x) \imp f(x) \sim_A g(x))
    \end{align*}
    is evidenceable. By $\pinc,\, \pprs,\, \psep(A)$ and $\pden(m)$, so is
    \begin{math}
        \prod_{x}(\rsex(x) \imp f(x) \sim_A g(x)).
    \end{math}
    Therefore, $\peq(\lift f, \lift g)$ is evidenceable and $f = g$ holds. \qed
\end{proof}

For the part (\ref{enumerate:main_result_2}) of the theorem, we need one more lemma.

\begin{lemma}
    \label{lem:appendix_4}
    $\prod_\phi \prod_\psi(j(\phi \imp \psi) \imp (j(\phi) \imp j(\psi)))$ is evidenceable.
\end{lemma}
\begin{proof}
    Take propositions $\phi, \psi \in \effomega$ arbitrarily, and assume that $j(\phi \imp \psi)$ and $j(\phi)$ are evidenceable.
    Then, we have an evidence of $j(\phi \imp \psi) \land j(\psi)$, and $j((\phi \imp \psi) \land \phi) \efequiv j(\psi)$ is also evidenceable by $\pprs$. \qed
\end{proof}

\begin{proof}[Theorem \ref{thm:main_result} (\ref{enumerate:main_result_2})]
    Suppose that $A$ is a $j$-sheaf.
    Let $g \colon S \to A$ be a morphism from a $j$-dense subobject $m \colon S \hookrightarrow X$, and let $\psi_x$ be the $j$-singleton for them.
    Since $m_A^*$ is bijective, there exists a unique extension $f \colon X \to A$ such that $g = f \circ m$.
    We take $a_x := f(x)$ and show that this witnesses $\pdsc(A)$,
    i.e., that for every existent $b \in A$, both of $\psi_x(b) \imp j(a_x \sim_A b)$ and $j(a_x \sim_A b) \imp \psi_x(b)$ are evidenceable.
    \begin{description}
        \item[$\psi_x(b) \imp j(a_x \sim_A b)$:] Assume that $\psi_x(b) = j\qty(\coprod_s (\rsex(s) \land m(s) \sim x \land g(s) \sim_A b))$ is evidenceable.
            If $m(s) \sim x$ holds for some existent $s$, then $f(m(s)) \sim_A f(x)$ by the extensionality of $f$, which is $g(s) \sim_A a_x$.
            Furthermore, with $g(s) \sim_A b$, $a_x \sim_A b$ is immediate. Thus, $\coprod_s (\rsex(s) \land m(s) \sim x \land g(s) \sim_A b) \imp a_x \sim_A b$ is evidenceable
            and $j(a_x \sim_A b)$ follows by $\pinc$ and Lemma \ref{lem:appendix_4}.
        \item[$j(a_x \sim_A b) \imp \psi_x(b)$:] Assume that $j(a_x \sim_A b) = j(f(x) \sim_A b)$ is evidenceable.
            Since $m$ is dense, $\coprod_s (\rsex(s) \land m(s) \sim x)$ is evidenceable and so is $j(\coprod_s (\rsex(s) \land m(s) \sim x))$ by $\pinc$.
            Combining this with $j(f(x) \sim_A b)$ and the fact $f(m(s)) \sim g(s)$, we have $j(\coprod_s (\rsex(s) \land m(s) \sim x \land f(x) \sim_A b))$,
            which is exactly $\psi_x(b)$.
    \end{description}

    Conversely, suppose that $A$ is $j$-separated and that $\pdsc(A)$ is evidenceable.
    We show that for every $j$-dense subobject $m \colon S \hookrightarrow X$ and every morphism $g \colon S \to A$,
    there exists a unique morphism $f \colon X \to A$ satisfying $g = f \circ m$.
    Let $\psi_x$ be the singleton for $m$ and $g$.
    By $\pdsc(A)$, for each $x \in X$,
    we can select an existent $a_x \in A$ making $\prod_b (\rsex(b) \imp (\psi_x(b) \efequiv j(a_x \sim_A b)))$ be evidenceable. Then, we define $f(x) := a_x$.
    For the final step, we need to give an evidence of
    \begin{align*}
        \peq(\lift{f \circ m}, \lift g) &= \prod_{s \in S} \prod_{b \in A} (\rsex(s, b) \imp (a_{m(s)} \sim_A b \efequiv g(s) \sim_A b)) \\
                                        &\efequiv \prod_{s \in S} (\rsex(s) \imp (a_{m(s)} \sim_A g(s))).
    \end{align*}
    For any existent $s \in S$, by $\rsex(g(s))$ we obtain
    \begin{align*}
        \psi_{m(s)}(g(s)) &= j\qty(\coprod_{t \in S}(\rsex(t) \imp m(s) \sim m(t) \land g(s) \sim_A g(t))) \\
                          &\efequiv j(a_{m(s)} \sim_A g(s)).
    \end{align*}
    Since the LHS is evidenceable by taking $t = s$, so is the RHS. Finally, $a_{m(s)} \sim_A g(s)$ is evidenceable by $\psep(A)$. \qed
\end{proof}

\section{Proofs for Section \ref{ssec:cpsca_vs_dnt_topos}}
\label{sec:appendix_proof_cpsca_dnt}

Prior to the proof, we state one of the axioms of the modality \cite[Def. 13]{cohen_partial_2025}.
The \emph{After-Return condition} requires that
$\phi(a) \le \modalangle{x}{\eta(a)}{\phi(x)}$ holds for any $\phi\colon A \to \Omega$ and any $a \in A$.

\begin{proof}[Lemma \ref{lem:k1_cps_extension}]
    Let $e$ be an evidence of $\phi$ in $\EF_\kleenek$ and $c \in \Natural$ be any code. Then, we have
    \begin{math}
        \htop \le \makleene{r}{\qty{\varphi_e(c)}}\phi(r)
    \end{math}
    and this immediately implies that $\varphi_e(c)$ is defined and that $\phi(\varphi_e(c)) = \htop$ holds.
    By the After-Return condition, we obtain
    \begin{math}
        \htop \le \macps{r}{[\varphi_e(c)]}\phi(r) = \macps{r}{e \cdot c}\phi(r).
    \end{math}
    The combinatory completeness of $\setpl$ ensures that this $e$ can be replaced by some $e' \in \setpl$.
    Hense, $\phi$ is evidenced by $e'$ in $\EF_\cps$.
    \qed
\end{proof}

\begin{proof}[Lemma \ref{lem:dne}]
    By the combinatory completeness of MCAs, there is a code $e$ satisfying $e \cdot c = \eta(\emem)$ for any code $c$ \cite[Def.\ 5]{cohen_partial_2025}. 
    We take $e$ as the evidence. Then, it suffices to prove $(\dnt\phi \imp \phi)(\emem) = \htop$ by the After-Return condition.
    Our goal is to show that for all $c \in \Natural$, if $\dnt\phi(c) = \htop$, then the following holds:
    \begin{equation}
        \label{eq:lem:dne}
        \macps{r}{\emem \cdot c}{\phi(r)} =
            \bigmeet_{k \in \Natural^{\Natural}}\qty(\bigmeet_{a \in \Natural}(\phi(a) \himp k \orthog a) \himp (\emem \cdot c) \orthog k) = \htop.
    \end{equation}
    Suppose $\dnt\phi(c) = \htop$, i.e.,
    \begin{align*}
        \bigmeet_{e \in \Natural}\qty((\lnot\phi)(e) \himp \macps{r}{c \cdot e} \orthog(r)) 
        = \bigmeet_{e \in \Natural}\qty((\lnot\phi)(e) \himp \bigmeet_{k' \in \Natural^{\Natural}} (c \cdot e) \orthog k') = \htop.
    \end{align*}
    Take an arbitrary continuous $k \colon \Natural \to \Natural$ and assume $\bigmeet_a (\phi(a) \himp k \orthog a) = \htop$.
    Then, it is enough to verify that $(\emem \cdot c) \orthog k$ holds.
    Again, we have the code $\ethrow{k}$ satisfying
    $(\ethrow{k} \cdot a)(k') = k(a)$ for all $a$ and all $k'$. This is an evidence of $\lnot \phi$:
    indeed,
    \begin{align*}
        (\lnot \phi)(\ethrow{k}) = \bigmeet_{a \in \Natural} \qty(\phi(a) \himp \bigmeet_{k' \in \Natural^{\Natural}}(\ethrow{k} \cdot a) \orthog k')
                                 = \bigmeet_{a \in \Natural} \qty(\phi(a) \himp k \orthog a)
    \end{align*}
    holds, and the RHS equals $\htop$ by assumption. Therefore, $\bigmeet_{k'} (c \cdot \ethrow{k})\orthog k' = \htop$.
    Recalling that $(\emem \cdot c)(k) = (c \cdot \ethrow{k})(k)$ holds, we obtain $(\emem \cdot c) \orthog k$ by taking $k' = k$ and Equation (\ref{eq:lem:dne}) follows. \qed
\end{proof}